\documentclass[12pt, 
]{article}
\usepackage{booktabs}

\usepackage[utf8]{inputenc}
\usepackage{amssymb}
\usepackage{amsmath}

\usepackage{geometry}
\newtheorem{theorem}{Theorem}

\newtheorem{remark}{Remark}

\newenvironment{proof}[1][Proof]{\emph{#1.} }{\  \hfill $\square $ \vspace{5 pt}}

\usepackage{natbib}
\usepackage{amssymb}
\usepackage[dvips]{graphics}
\usepackage{amsmath}

\usepackage{mathpazo}
\usepackage{enumerate}

\usepackage{amsfonts}

\usepackage{amssymb}

\usepackage{enumerate}

\usepackage{mathrsfs}

\usepackage[usenames]{color}
\usepackage{cancel}

\usepackage{pgf,tikz}

\usetikzlibrary{decorations.pathreplacing,angles,quotes}
\usetikzlibrary{arrows.meta}
\usetikzlibrary{arrows, decorations.markings}
\tikzset{myptr/.style={decoration={markings,mark=at position 1 with %
       {\arrow[scale=2,>=stealth]{>}}},postaction={decorate}}}

\usepackage{hyperref}
\hypersetup{
    colorlinks,
    citecolor=blue,
    filecolor=black,
    linkcolor=blue,
    urlcolor=blue
}

\usepackage{pgf,tikz}
\usetikzlibrary{arrows}

\usepackage{fancyhdr}

\usepackage{soul}

\usepackage{emptypage}

\newcommand*\samethanks[1][\value{footnote}]{\footnotemark[#1]}

\usepackage[T1]{fontenc}
\DeclareFontFamily{T1}{calligra}{}
\DeclareFontShape{T1}{calligra}{m}{n}{<->s*[1.44]callig15}{}
\DeclareMathAlphabet\mathcalligra   {T1}{calligra} {m} {n}

\usepackage{multirow}
\usepackage{array}
\usepackage{mathtools}
\usepackage{arydshln}
\usepackage{threeparttable}
\usepackage{booktabs,caption}

\usepackage{ifthen}
\newboolean{showcomments}
\setboolean{showcomments}{true}

\newcommand{\pablo}[1]{  \ifthenelse{\boolean{showcomments}}
{\textcolor{green!50!black}{(T: #1)}}{}}
\newcommand{\marcelo}[1]{\ifthenelse{\boolean{showcomments}}
{\textcolor{red}{(M: #1)}}{}}
\newcommand{\agustin}[1]{  \ifthenelse{\boolean{showcomments}}
{\textcolor{blue!50!black}{(T: #1)}}{}}

\usepackage{orcidlink}  

\begin{document}

\title{Obvious manipulations, consistency, and the uniform rule%
\thanks{%
We acknowledge financial support
from UNSL through grants 032016, 030120, and 030320, from Consejo Nacional
de Investigaciones Cient\'{\i}ficas y T\'{e}cnicas (CONICET) through grant
PIP 112-200801-00655, and from Agencia Nacional de Promoción Cient\'ifica y Tecnológica through grant PICT 2017-2355.}}


\author{R. Pablo Arribillaga\thanks{
Instituto de Matem\'{a}tica Aplicada San Luis (UNSL and CONICET) and Departamento de Matemática, Universidad Nacional de San
Luis, San Luis, Argentina. Emails: \href{mailto:rarribi@unsl.edu.ar}{rarribi@unsl.edu.ar} 
and \href{mailto:abonifacio@unsl.edu.ar}{abonifacio@unsl.edu.ar} 
} \hspace{1.5 pt} \orcidlink{0000-0002-0521-0301} \and Agustín G. Bonifacio\samethanks[2] 
\hspace{1.5 pt} \orcidlink{0000-0003-2239-8673}}

\date{\today}

\maketitle

\begin{abstract}

In the problem of fully allocating an infinitely divisible commodity among agents whose preferences are single-peaked, we show that the uniform rule is the only allocation rule that satisfies \emph{efficiency}, the \emph{equal division guarantee}, \emph{consistency}, and \emph{non-obvious manipulability}.

\bigskip

\noindent \emph{JEL classification:} D51, D63, D70. \bigskip

\noindent \emph{Keywords:} obvious manipulations, consistency, uniform rule, allotment rules.  

\end{abstract}

\section{Introduction}

Consider the allocation of a single, non-disposable, and infinitely divisible commodity among a group of agents with single-peaked preferences: up to some critical level, known as the \emph{peak}, an increase in consumption enhances an agent's welfare, while beyond that level, the opposite holds. 

A \emph{rule} in this context refers to a systematic procedure that determines an allocation for each possible problem of this kind. The mechanism design literature aims to identify rules that are well behaved with respect to various desiderata, such as efficiency, consistency, fairness, and incentive compatibility. Among the many rules proposed, one stands out for its ability to satisfy a wide range of characterizations: the \emph{uniform rule}. Introduced in the axiomatic literature by \cite{sprumont1991division}, the uniform rule selects an efficient allocation that is as close to equal division as possible: each agent receives either their peak amount or a common amount chosen to ensure feasibility of the overall assignment. 

Of paramount importance is the fact that the uniform rule is \emph{strategy-proof}, i.e., manipulations of the rule by misreporting preferences are not possible. Several characterizations of the rule involve strategy-proofness \citep[see, among others,][]{sprumont1991division,ching1994alternative}. Recently, a weakening of the notion of strategy-proofness has been proposed. It relies on a refinement of the concept of manipulation. Assuming that an agent knows all possible outcomes under any preference misreport, a manipulation (i.e. a profitable  misreport) is \emph{obvious} if it either makes the agent better off than truth-telling in the worst case or makes the
agent better off than truth-telling in the best case, and a rule is \emph{non-obviously manipulable} if no manipulation of the rule results obvious.   The idea underlying this requirement, introduced by \cite{troyan2020obvious}, is that even though manipulations are pervasive, agents may not realize they can manipulate a rule because they lack information about others’ behavior or they are cognitively limited. 


Relaxing strategy-proofness to this weaker version enables the exploration of a diverse range of compelling allocation rules.  \cite{arribillaga2023not} identify a broad class of non-obviously manipulable rules, among which the uniform rule is included.\footnote{The existence of an abundance of non-obviously manipulable rules has been established in various settings, including two-sided matching \citep{troyan2020obvious,arribillaga2023obvious}, voting \citep{aziz2021obvious,arribillaga2024obvious}, and cake-cutting \citep{ortega2022obvious}.} This raises the question of whether this central rule can be uniquely characterized by imposing, together with non-obvious manipulability, additional meaningful properties. One such property is \emph{consistency}, a variable-population invariance criterion that has played a fundamental role in axiomatic analysis \citep{thomson2012axiomatics}. Consistency says that the allocation recommended by a rule in a given problem should remain in agreement with the allocation it would recommend in a reduced problem obtained when some participants leave with their assigned allotments. 


 


In this note, we propose a new characterization of the uniform rule marrying non-obvious manipulability, consistency, efficiency, and a mild fairness axiom. This last one, called the \emph{equal division guarantee}, requires that if an agent’s most preferred level of consumption matches equal division, they should receive exactly that amount.


There are characterizations of the uniform rule involving consistency \citep[see, among others,][]{thomson1994consistent, sonmez1994consistency}. However, to the best of our knowledge, the result we present here is the first characterization of this rule that incorporates an incentive compatibility criterion alongside consistency.

\section{Preliminaries}


We consider the set of natural numbers $\mathbb{N}$ as the set of  \textbf{potential agents.} Denote by $\mathcal{N}$ the collection of all finite subsets of   $\mathbb{N}.$ Each agent $i \in \mathbb{N}$ is equipped with a \textbf{single-peaked} preference relation $R_i$ defined over $\mathbb{R}_+$ with a unique peak $p(R_i)$, and such that for each pair $\{x_i, x_i'\} \subseteq \mathbb{R}_+$, we have $x_iP_ix_i'$ as long as either $x_i'<x_i\leq p(R_i)$ or $p(R_i) \leq x_i<x_i'$ holds. Denote by $\mathcal{SP}$ the domain of all such preferences. Call $P_i$ and $I_i$ to the strict preference and indifference relations associated with $R_i,$ respectively. 

Given $N \in \mathcal{N},$ an  \textbf{economy} 
consists of a profile of preferences $R=(R_j)_{j \in N} \in \mathcal{SP}^N$ and a social  endowment $\Omega \in \mathbb{R}_{++}$ and is
denoted by $(R,\Omega)$. Let  $\mathcal{E}^N$ be the domain of economies with agents in $N$. Let  $\mathcal{E}=\bigcup_{N \in \mathcal{N}} \mathcal{E}^{N}$ denote the set of all potential economies. For each $N \in \mathcal{N}$ and each $(R, \Omega) \in \mathcal{E}^N,$ let $X(R, \Omega)=\{x \in  \mathbb{R}^{N}_+ : \sum_{j\in N}x_j= \Omega\}$ be the set of \textbf{allocations} for economy $(R, \Omega)$ and let $X=\bigcup_{(R, \Omega) \in \mathcal{E}}X(R, \Omega).$ A  \textbf{rule} is a function $\varphi: \mathcal{E} \longrightarrow X$ such that $\varphi(R, \Omega) \in X(R, \Omega)$ for each $(R, \Omega) \in \mathcal{E}.$


Given a rule $\varphi$, some desirable properties we consider are listed next. 

The first property demands the rule to always recommend a Pareto optimal allocation. In this model, this is equivalent to asking that, in economies with excess demand (the sum of the peaks surpasses the social endowment), each agent obtains at most their peak amount; whereas in economies with excess supply  (the sum of the peaks falls behind the social endowment), each agent receives at least their peak amount. To formalize this, for each $N \in \mathcal{N}$ and each economy $(R,\Omega) \in \mathcal{E}^N,$  let $z(R, \Omega)=\sum_{j \in N}p(R_j)-\Omega.$ 

\vspace{5 pt}

\noindent
\textbf{Efficiency:} For each $N \in \mathcal{N}$ and each $(R, \Omega) \in \mathcal{E}^N$,  $z(R, \Omega) \geq 0$ implies  $\varphi_i(R, \Omega)\leq p(R_i)$ for each $i \in N,$
and  $z(R, \Omega)\leq 0$ implies  $\varphi_i(R, \Omega)\geq p(R_i)$ for each $i \in N.$
\vspace{5 pt}

The second property, introduced in \cite{arribillaga2023not}, is a mild fairness requirement. It says that an agent should obtain their share of the equal division whenever their peak consumption demands it.

\vspace{5 pt}
\noindent
\textbf{Equal division guarantee:} For each $N \in \mathcal{N}$, each  $(R, \Omega) \in \mathcal{E}^N$ and each $i \in N$ such that $p(R_i)=\frac{\Omega}{|N|}$, we have $\varphi_i(R, \Omega)=\frac{\Omega}{|N|}$.

\vspace{5 pt}

The third property requires coherence in the treatment of economies and their associated sub-economies. Specifically, if some agents leave an economy with their allocations as determined by the rule, the rule must assign to each remaining agent in the sub-economy the same allocation as in the original economy. 
Let $N \in \mathcal{N}$. If $N' \subseteq N$ and $R \in
\mathcal{SP}^N,$ let $R_{N'}=(R_j)_{j \in N'}$ denote the restriction of
$R$ to $N'.$ 

\vspace{5 pt}

\noindent
\textbf{Consistency:} For each pair $N, N' \in \mathcal{N}$ such that $N' \subseteq N$, each $(R, \Omega) \in \mathcal{E}^N$, and each $i \in N'$,  we have  $\varphi_i(R_{N'}, \sum_{j \in N'}\varphi_j(R, \Omega))=\varphi_i(R, \Omega)$.

\vspace{5 pt}


Finally, for our incentive compatibility property, we need some definitions. Let $N\in \mathcal{N}$,  $i\in N$, $R_i\in \mathcal{SP}$ and $\Omega \in \mathbb{R}_{++}$.
The \textbf{option set attainable with $\boldsymbol{(R_{i}, \Omega)$ at $\varphi}$}  is 
\begin{equation*}
O^\varphi(R_{i},\Omega)=\left\{x \in [0, \Omega] 
\ : \ x=\varphi_i(R_{i},R_{N\setminus\{i\}}) \text{ for some } R_{N\setminus\{i\}}\in \mathcal{SP}^{N\setminus\{i\}}\right\}.\footnote{\cite{barbera1990strategy} were the first to use option sets in the context of preference aggregation.}
\end{equation*}
Preference   $R_i' \in \mathcal{SP}$ is an \textbf{obvious manipulation  of  
$\boldsymbol{\varphi$ at $(R_i,\Omega)}$} if:
\begin{enumerate}[(i)]
    \item there is  $R_{N\setminus\{i\}} \in \mathcal{SP}^{N\setminus\{i\}}$ such that 
$\varphi_i(R_i', R_{N\setminus\{i\}}, \Omega) \ P_i \ \varphi_i(R_i, R_{N\setminus\{i\}}, \Omega)$; and

\item for each 
$x' \in O^\varphi(R_{i}^{\prime },\Omega)$ there is $x \in O^\varphi(R_{i},\Omega)$ such that $x'P_i x.$

\end{enumerate}
When $R_i'$ satisfies (i) we say that $R_i'$ is a \textbf{manipulation of $\boldsymbol{\varphi$ at $(R_i, \Omega)}$}. A manipulation becomes obvious if (ii) holds, i.e., if each possible outcome under the manipulation is strictly better than some possible outcome under truth-telling.


The following property precludes such behavior. 
\vspace{5 pt}

\noindent \textbf{Non-obvious manipulability:} For each $N \in \mathcal
N$, each  $i \in N$, and each  $(R_i, \Omega) \in \mathcal{SP} \times \mathbb{R}_{++},$ there is no obvious manipulation of $\varphi$ at $(R_i, \Omega)$.  

\vspace{5 pt}

If rule $\varphi$ has no manipulation, it is \textbf{strategy-proof}. Clearly, \emph{strategy-proofness} implies \emph{non-obvious manipulability}.

\begin{remark}     
The original definition of obvious manipulation introduced by \cite{troyan2020obvious} says that the best possible outcome in $O^\varphi(R_{i}^{\prime },\Omega)$ should be strictly better than the best possible outcome in $O^\varphi(R_{i},\Omega)$ or the worst possible outcome in $O^\varphi(R_{i}^{\prime },\Omega)$ should be strictly better than the worst possible outcome in $O^\varphi(R_{i},\Omega)$.  
However, under efficiency, the best possible outcome under truth-telling for an agent is always the agent's peak alternative in our model and thus no manipulation becomes obvious by considering such outcomes. Furthermore, when the set of alternatives is infinite (as in our case), worst possible outcomes may not be well-defined and a more general definition is necessary. For a formal discussion and comparison of our definition of \emph{non-obvious manipulability} and the original presented by \cite{troyan2020obvious}, see \cite{arribillaga2023not}.
\end{remark}

 The rule that has played a predominant role in this literature is the uniform rule:

\vspace{10 pt}

\noindent \textbf{Uniform rule, $\boldsymbol{u}$:} For each $N \in \mathcal{N},$ each $ (R, \Omega)\in \mathcal{E}^N,$ and each $i \in N,$
$$u_i(R, \Omega)=\left\{\begin{array}{l l }
\min\{p(R_i), \lambda\} & \text{if } z(R, \Omega)\geq 0\\
\max\{p(R_i), \lambda\} & \text{if } z(R,\Omega)< 0\\
\end{array}\right.$$
where $\lambda \geq 0$ and solves $\sum_{j \in N}u_j(R, \Omega)=\Omega.$

\vspace{10 pt}

\section{Characterization} 

Now we present our characterization result. 

\begin{theorem}\label{theo main}
The uniform rule is the only rule 
that satisfies efficiency, the equal division guarantee, consistency, and non-obvious manipulability. 
\end{theorem}
\begin{proof}
It is clear that the uniform rule satisfies \emph{efficiency}, the \emph{equal division guarantee}, \emph{consistency}  \citep[][]{thomson1994consistent}, and \emph{strategy-proofness} \citep[][]{sprumont1991division} implying therefore \emph{non-obvious manipulability}.
 
Now, let us see uniqueness.  Let $\varphi: \mathcal{E} \longrightarrow X$ be a rule satisfying \emph{efficiency}, the \emph{equal division guarantee}, \emph{consistency}, and \emph{non-obvious manipulability}.  We will prove that $\varphi=u$. 
Let $N \in \mathcal{N}$ and $ (R,\Omega) \in \mathcal{E}^N$. The case $\sum_{j \in N} p(R_j) =\Omega$ is trivial by \emph{efficiency}. Without loss of generality, assume that $\sum_{j \in N} p(R_j) > \Omega.$ The case $\sum_{j \in N} p(R_j) <\Omega$ can be handled symmetrically. Let $|N|=n$. The proof proceeds in three steps:

\medskip

\noindent \textbf{Step 1: If  $\boldsymbol{p(R_i) \geq \frac{\Omega}{n}$ for each $i \in N}$,} \textbf{then  $\boldsymbol{\varphi_i(R,\Omega)=\frac{\Omega}{n}}$ for each $\boldsymbol{i \in N}$}. 
We will prove that $\varphi_i(R, \Omega) \geq \frac{\Omega}{n}$ for each $i \in N$. Thus, by feasibility,  $\varphi_i(R, \Omega) = \frac{\Omega}{n}$ for each $i \in N$. Assume on contradiction that there is $i\in N$ such that $\varphi_i(R, \Omega) < \frac{\Omega}{n}$. Then, by single-peakedness, $\frac{\Omega}{n} P_i \varphi_i(R, \Omega)$. Now, let $R_i'\in \mathcal{SP}$ be such that $p(R'_i) = \frac{\Omega}{n}$. By the \emph{equal division guarantee}, $\varphi_i(R'_i,R_{N \setminus \{i\}}, \Omega)=\frac{\Omega}{n}$. Therefore, $R_i'$ is a manipulation of $\varphi$ at $(R_i, \Omega)$. Furthermore, as $O^\varphi(R_i'))=\{\frac{\Omega}{n}\}$, we have that $R_i'$ is an obvious manipulation of $\varphi$ at $(R_i, \Omega)$. This contradicts that  $\varphi$ satisfies \emph{non-obvious manipulability}. Hence, $\varphi_i(R,\Omega)=\frac{\Omega}{n}$ for each $i \in N$.

\medskip

\noindent \textbf{Step 2: Let $\boldsymbol{i \in N}.$ If $\boldsymbol{p(R_i) \leq \frac{\Omega}{n},$ then $\varphi_i(R, \Omega)=p(R_i)}$}.  By \emph{efficiency}, $\varphi_i(R, \Omega) \leq p(R_i).$ Assume, on contradiction, that  $\lambda \equiv \varphi_i(R, \Omega)<p(R_i).$ Let $\gamma \in (\lambda,p(R_i)).$ Then, $p(R_i)-\gamma>0$. First, choose $k \in \mathbb{N}$ such that 
 \begin{equation} \label{nomj}
    k(p(R_i)-\gamma)>\Omega-n p(R_i).
\end{equation} 
Since $\gamma<p(R_i) \leq \frac{\Omega}{n}$, we have $n\gamma < \Omega$. Then,
\begin{equation}\label{gamamenor}
    \gamma=
    \frac{n \gamma+k\gamma}{n+k}<\frac{\Omega+k\gamma}{n+k}.
\end{equation}
Therefore, by \eqref{nomj} and \eqref{gamamenor}, 
\begin{equation}\label{egregium}
    \gamma < \frac{\Omega+k\gamma}{n+k} < p(R_i).
\end{equation}
Now, let $N^\star \in \mathcal{N}$ be such that $N \subseteq N^\star$ and $|N^\star|=n+k$, and let $\Omega^\star=\Omega+k\gamma$. Consider an economy $\left(R^\star,\Omega^\star\right) \in \mathcal{E}^{N^\star}$  where $R^\star_N=R$ and, for each $j \in N^\star \setminus N$, $R_j^\star \in \mathcal{SP}$ is such that  
\begin{equation} \label{nomguaranted}
    p(R_j^\star)=\gamma \text{ \ and \  }\frac{\Omega^\star}{n+k} \ P_j^\star \ x \text{ \ for each \ }x \in (0,\gamma).
\end{equation}
The existence of such $R_j^\star \in \mathcal{SP}$ is guaranteed by \eqref{egregium}.

Since $\sum_{j \in N^\star}p(R_j^\star)=\sum_{j \in N}p(R_j)+\sum_{j \in N^\star\setminus N}p(R_j^\star)>\Omega+k\gamma=\Omega^\star,$ it follows that $\sum_{j \in N^\star}p(R_j^\star)>\Omega^\star.$ By \emph{efficiency}, we have  $\varphi_j(R^\star,\Omega^\star)\leq p(R_j^\star)=\gamma$ for each $j\in N^\star \setminus N.$ 

\smallskip

\noindent \textbf{Claim: $\boldsymbol{\varphi_j(R^\star,\Omega^\star)= \gamma}$ for each $\boldsymbol{j\in N^\star \setminus N}.$} Otherwise, $\varphi_j(R^\star,\Omega^\star)< \gamma$ for some $j\in N^\star \setminus N$ implies, by the \emph{equal division guarantee} and \eqref{nomguaranted}, that such agent $j$ has an obvious manipulation of $\varphi$ at $(R_j^\star, \Omega^\star)$ by declaring a preference with $\frac{\Omega^\star}{n+k}$  as its peak alternative. This contradicts that $\varphi$ is \emph{non-obvious manipulability}, so the claim is proved.

\smallskip 

\noindent Therefore, by the Claim we have $\sum_{j \in N^\star\setminus N}\varphi(R_j^\star, \Omega^\star)=k \gamma$ or, equivalently,  $\sum_{j \in N}\varphi(R_j^\star, \Omega^\star)=\Omega$. Thus, by \emph{consistency},  $\varphi_i(R^\star,\Omega^\star)= \varphi_i(R,\Omega)=\lambda$. As $\lambda< \gamma$, \eqref{egregium} implies $$\varphi_i(R^\star, \Omega^\star)<\frac{\Omega^\star}{n+k}<p(R_i)=p(R_i^\star),$$ so agent $i$, by the \emph{equal division guarantee}, has an obvious manipulation of $\varphi$ at $(R^\star_i, \Omega^\star)$ by declaring a preference with  $\frac{\Omega^\star}{n+k}$  as its peak alternative. This contradicts that $\varphi$ is \emph{non-obvious manipulability}. Hence, $\varphi_i(R,\Omega)=p(R_i)$, as desired.

\medskip

\noindent \textbf{Step 3: Concluding}. Without loss of generality, assume $N=\{1,2,\ldots, n\}$ and $p(R_1)\leq p(R_2) \leq \ldots \leq p(R_n)$. There are two cases to consider:

\begin{itemize}
    \item[$\boldsymbol{1}.$] \textbf{$\boldsymbol{p(R_i)\geq \frac{\Omega}{n}$ for each $i \in N}$}.  Then, by Step 1, $\varphi_i(R, \Omega)=\frac{\Omega}{n}=u_i(R, \Omega)$ for each $i \in N$. 
    \item[$\boldsymbol{2}.$] $\boldsymbol{p(R_1)<\frac{\Omega}{n}}.$ Then, by Step 2, $\varphi_1(R, \Omega)=p(R_1)=\min\{p(R_1), \frac{\Omega}{n}\}=u_1(R, \Omega)$. Now, let agent 1 leave with their consumption. By \emph{consistency},  
    \begin{equation}\label{eq1 charact}
    \varphi_2(R, \Omega)=\varphi_2(R_{N \setminus \{1\}}, \Omega - \varphi_1(R, \Omega)).  
    \end{equation}
    By a similar argument to the one used in economy $(R, \Omega)$ for agent 1 and the fact that $\varphi_1(R, \Omega)=u_1(R, \Omega)$, it follows that  
    \begin{equation}\label{eq2charact}
    \varphi_2(R_{N \setminus \{1\}}, \Omega - \varphi_1(R, \Omega))=\min\left\{p(R_2), \frac{\Omega- u_1(R, \Omega)}{n-1}\right\}=u_2\left(R_{N \setminus \{1\}}, \Omega - u_1(R, \Omega)\right).  
    \end{equation}
    As $u$ is \emph{consistent}, $u_2(R_{N \setminus \{1\}}, \Omega - u_1(R, \Omega))=u_2(R, \Omega)$ and, using \eqref{eq1 charact} and \eqref{eq2charact}, it follows that $\varphi_2(R,\Omega)=u_2(R, \Omega).$ Continuing in the same way, letting agents leave the economy one by one with their consumptions, we obtain the desired result.\end{itemize}
Hence, it follows that $\varphi(R,\Omega)=u(R, \Omega).$ 
\end{proof}

\begin{remark}
    The reasoning in Step 3 of the previous theorem closely follows the approach outlined in Theorem 1 of \cite{sonmez1994consistency}.
\end{remark}

To analyze the independence of the axioms involved in the characterization of Theorem \ref{theo main}, we consider the following rules.

\begin{itemize}
     \item Let $\widetilde{\varphi}:\mathcal{E} \longrightarrow X$ be such that, for each $N\in \mathcal{N}$ and each $(R,\Omega)\in \mathcal{E}^N$ and each $i\in N$, $\varphi_i(R,\Omega)=\frac{\Omega}{|N|}$ . Then $\widetilde{\varphi}$ satisfies all properties but \emph{efficiency}. 

     \item Let $\varphi^<:\mathcal{E} \longrightarrow X$ be a sequential dictator. For each $N\in \mathcal{N}$ and each $(R,\Omega)\in \mathcal{E}^N$, agents in $N$ following the usual order $<$ choose their best option within the amount left by previous agents. Then, $\varphi^<$ satisfies all properties but the \emph{equal division guarantee}. 

     \item Let $\widetilde{\varphi}:\mathcal{E} \longrightarrow X$ be such that, for each $N\in \mathcal{N}$ and each $(R,\Omega)\in \mathcal{E}^N$,   
    $$\overline{\varphi}(R, \Omega)=\left\{\begin{array}{l l }
    (\frac{1}{3}\Omega, \frac{2}{3}\Omega, 0, \ldots,0) & \text{if } p(R_1)=p(R_2)=\Omega \text{ and } \\ & p(R_j) = 0 \text{ for each }j \in N \setminus\{1,2\} \\
    u(R, \Omega) & \text{otherwise}\\
    \end{array}\right.$$

This rule satisfies \emph{efficiency} and the \emph{equal division guarantee}. Observe that $$O^{\overline{\varphi}}(R_i,\Omega)=O^u(R_i,\Omega)=\begin{cases}
    \left[\frac{\Omega}{n}, p(R_i)\right] & \text{ if } p(R_i)>\frac{\Omega}{n} \\
    \left[p(R_i), \frac{\Omega}{n}\right] & \text{ if } p(R_i)\leq \frac{\Omega}{n}\\
\end{cases}$$
Therefore, if $R'_i$ is a manipulation of $\overline{\varphi}$ at $(R_i, \Omega)$, $\frac{\Omega}{n}\in O^{\overline{\varphi}}(R'_i,\Omega)$. Then, by single-peakedness, there is no $x\in O^{\overline{\varphi}}(R_i,\Omega)$  such that $\frac{\Omega}{n}P_ix$, implying that  $\overline{\varphi}$ satisfies \emph{non-obvious manipulability}. 
Therefore, $\overline{\varphi}$ satisfies all properties but \emph{consistency}.

\item Let $\varphi^\star:\mathcal{E} \longrightarrow X$ be such that, for each $N\in \mathcal{N}$ and each $(R,\Omega)\in \mathcal{E}^N$,   
    $$\varphi^\star(R, \Omega)=\left\{\begin{array}{l l }
    \varphi^m(R, \Omega) & \text{if } z(R,\Omega)> 0 \\
    u(R, \Omega) & \text{otherwise}\\
    \end{array}\right.$$

Where $\varphi^m:\mathcal{E} \longrightarrow X$ is a sequential dictator where the order of the agents is given by ordering their peaks in an increasing manner.\footnote{Ties are broken by an arbitrary fixed order.} For each $N\in \mathcal{N}$ and each $(R,\Omega)\in \mathcal{E}^N$, agents in $N$ following the order of their preferences' peaks choose their best option within the amount left by previous agents. Clearly, $\varphi^\star$ satisfies \emph{efficiency} and the \emph{equal division guarantee}. As $\varphi^m$ and $u$ satisfy \emph{consistency} it follows that $\varphi^\star$ also satisfies it. Then, $\varphi^\star$ satisfies all properties but \emph{non-obvious manipulability}.

\end{itemize} 

    


\bibliographystyle{ecta}
\bibliography{biblio}

\end{document}